\documentclass[11pt]{article} 

\usepackage{fullpage} 
\usepackage{bbm}
\usepackage{graphicx}
\usepackage{upref}
\usepackage{enumerate}
\usepackage{latexsym}
\usepackage{color,graphics}
\usepackage{comment} 
\usepackage{relsize}

\usepackage[section]{algorithm} 
\usepackage{amsmath,amsthm,amssymb,algorithmic,epsfig,graphicx, tikz}

\usepackage{amsfonts}
\usepackage{varioref}
\usepackage[ansinew]{inputenc}

\usepackage{amsmath}
\usepackage{amsfonts}
\usepackage{tikz} 

\usepackage{amssymb}
\usepackage{varioref}
\usepackage[ansinew]{inputenc}

\newtheorem{theorem}{Theorem}[section]
\newtheorem{claim}[theorem]{Claim}
\newtheorem{lemma}[theorem]{Lemma}
\newtheorem{corollary}[theorem]{Corollary}
\newtheorem{proposition}[theorem]{Proposition}
\newtheorem{fact}[theorem]{Fact}

\newtheorem{defn}[theorem]{Definition}

\newtheorem*{thm:composite}{Theorem \ref{thm:composite}}
\newtheorem*{thm:large}{Theorem \ref{thm:large}}

\def\Z{{\mathbb Z}}

\def\C{{\mathbb C}}
\def\F{{\mathbb F}}

\newcommand{\ket}[1]{{|{#1}\rangle}}

\newcommand{\poly}{\mathrm{poly}}

\newcommand\supp{{\mathrm{supp}}}

\newcommand{\HMS}{\mbox{\rmfamily\textsc{HMS}}}

\newcommand{\LPN}{\mbox{\rmfamily\textsc{LPN}}}
\newcommand{\LPSN}{\mbox{\rmfamily\textsc{LPSN}}}
\newcommand{\LWE}{\mbox{\rmfamily\textsc{LWE}}}
\newcommand{\LFD}{\mbox{\rmfamily\textsc{LFD}}}
\newcommand{\LFS}{\mbox{\rmfamily\textsc{LFS}}}
\newcommand{\HT}{\mbox{\rmfamily\textsc{HS}}}

\def\ket#1{\mathinner{|{#1}\rangle}}

\newcommand{\norm}[1]{\left\lVert#1\right\rVert}

\newcommand{\suppress}[1]{}

\title{On learning linear functions from subset and its applications in quantum computing}

\author{ 
G\'abor Ivanyos \thanks{ Institute for Computer Science and Control, Hungarian Academy of Sciences, Budapest, Hungary, ({\tt Gabor.Ivanyos@sztaki.mta.hu})} 
\and 
Anupam Prakash \thanks{ CNRS, IRIF, Universit\'e Paris Diderot 75205 Paris, France;  ({\tt anupam@irif.fr})} 
\and 
Miklos Santha \thanks{ CNRS, IRIF, Universit\'e Paris Diderot 75205 Paris, France;  and Centre for Quantum Technologies, National University of Singapore, 
Singapore 117543 ({\tt miklos.santha@gmail.com}).} 
}

\begin{document}

\maketitle

\begin{abstract}
Let $\F_{q}$ be the finite field of size $q$ and let $\ell: \F_{q}^{n} \to \F_{q}$ be a linear function. 
We introduce the {\em Learning From Subset} problem $\LFS(q,n,d)$ 
of learning $\ell$, given samples $u \in \F_{q}^{n}$ from a special distribution
depending on $\ell$: the probability of sampling $u$ is a function of $\ell(u)$ and is non zero for at most $d$ values of $\ell(u)$. 
We provide a randomized algorithm for $\LFS(q,n,d)$ with sample complexity $(n+d)^{O(d)}$ and running time 
polynomial in $\log q$ and $(n+d)^{O(d)}$. Our algorithm generalizes and improves upon previous results 
\cite{Friedl, Ivanyos} that had provided algorithms for $\LFS(q,n,q-1)$ with running time $(n+q)^{O(q)}$. 
We further present applications of our result to the {\em Hidden Multiple Shift} problem $\HMS(q,n,r)$ in quantum computation 
where the goal is to determine the hidden shift $s$ given oracle access to $r$ shifted copies of an injective function $f: \Z_{q}^{n} \to \{0, 1\}^{l}$, 
that is we can make queries of the form $f_{s}(x,h) = f(x-hs)$ where $h$ can assume $r$ possible values. We reduce $\HMS(q,n,r)$ to $\LFS(q,n, q-r+1)$ to obtain a
polynomial time algorithm for $\HMS(q,n,r)$ when $q=n^{O(1)}$ is prime and $q-r=O(1)$. The best known algorithms \cite{CD07, Friedl} for $\HMS(q,n,r)$ with these 
parameters require exponential time. 
 \end{abstract}
 
\section{Introduction} 
\subsection{Learning with noise}
Let $n\geq 1$ and $q > 1$ be integers.
We denote by $\Z_q$ 
the ring of integers modulo $q$, and
by $\F_q$ the finite field on $q$ elements,
when $q$ is some power of a prime number. When $q$ is prime then $\Z_q$ coincides with $\F_q$, and 
we will use the notation $\F_q$.
Let  $\ell: \F_q^n \rightarrow \F_q$ be an 
$n$-variable linear function. The main subject of this paper is to learn $\ell$ given partial information about 
the values $\ell(u)$ for uniformly random samples $u$ from $\F_q^n$. 
In the ideal setting, when we 
have access to the  values $\ell(u)$ for uniformly random samples from $\F_q^n$, the problem is canonical and
perfectly understood: after getting $n$ independent samples, we can determine $\ell$ by Gaussian
elimination in polynomial time.
But when instead of the exact values we receive only some property satisfied by them, the problem can become
much more difficult.

Since an element of $\F_q^n$ can be specified with $n \log q$ bits, we will say that {\em an algorithm 
is in polynomial time}
if it runs in time polynomial in both $n$ and $\log q$. Let $f(n,q)$ be a function of $n$ and $q$, then we say that a function $g(n,q) \in \widetilde{O}(f)$ 
if $g(n,q)\leq f(n,q) \log^{c}(nq)$ for some constant $c$ for sufficiently large $n$ and $q$.  By the {\em sample complexity} of an algorithm we mean the number of samples used by it.

There is a somewhat similar context to the learning model we investigate,
it is the model where the values $\ell(u)$ are perturbed by some random noise.
The first example of such a work is by Blum et al.~\cite{Blum} 
on the {\em Learning Parity with Noise} problem $\LPN(n, \eta)$, where 
$\eta < 1/2$.
Here we have access to tuples $(u,b) \in \F_2^n \times \F_2$, where $u$ is a uniformly random element of $\F_2^n$
and $b = \ell(u) + e$, where $e$ is a random 0--1 variable with $\Pr [ e = 1] = \eta$.
For constant noise rate $0 < \eta <1/2$, the best known algorithm for $\LPN(n, \eta)$ is from~\cite{Blum}. 
It has both sample and time complexity of $2^{O(n/ \log n)}$, and therefore only marginally beats
the trivial exhaustive search algorithm of complexity $2^{O(n)}$.

The {\em Learning With Error} problem $\LWE(q, n, \chi)$ is a generalization by Regev \cite{Regev} 
of $\LPN$ to larger fields.
Here $q$ can be any prime number, and $\chi$ is a probability distribution on $\F_q$.
Similar to $\LPN$, we have access to tuples $(u,b) \in \F_q^n \times \F_q$, where $u$ is a uniformly random element of $\F_q^n$ and $b = \ell(u) + e$, with the random variable $e$ 
having distribution $\chi$.
Under the assumptions that $q$ is bounded by some polynomial function of $n$, and that
$\chi(0) \geq 1/q + 1/p(n)$, for some polynomial $p$, the problem can be solved classically with sample and time complexity
$2^{O(n)}$. The case when $\chi=\Psi_{\alpha}$, the discrete Gaussian distribution of standard deviation
$\alpha q$, is of particular interest for lattice based cryptography. Indeed, one of the main results of~\cite{Regev}
is that for appropriate parameters, solving $\LWE(q, n, \Psi_{\alpha})$ is at least as hard as
quantumly solving several cryptographically important lattice problems in the worst case. In a subsequent
work a classical reduction of some of these lattice problems to LWE was given by Peikert~\cite{Peikert}.

In~\cite{AG} Arora and Ge introduced a more structured noise model for learning linear functions
over $\F_2^n$. In the {\em Learning Parity with Structured Noise} problem $\LPSN(n,m)$ the samples arrive 
in groups of size $m$, that is in one sampling step we receive $(u_1,b_1), \ldots , (u_m,b_m)$,
where $(u_i,b_i) \in \F_2^n \times \F_2$, for $i = 1, \ldots, m$.
Here $u_1, \ldots , u_m$ are independent random elements drawn from $\F_2^n$, and
$b_i = \ell(u_i) +e_i$, where the 
the noise vector $e =(e_1, \ldots , e_m) \in \F_2^m$ must have Hamming weight less than $m/2$.
The chosen noise vector $e$ can depend on the sample $(u_1, \ldots , u_m)$, but the model has an important
restriction (structure) compared to the previous error models. Since the Hamming weight of $e$ is less than $m/2$,
it is guaranteed that in every sampling group the majority of the bits $b_i$ is correct, that is coincides with
$\ell(u_i)$. In fact, the model of Arora and Ge is somewhat more general. 
Let $P$ be any $m$-variable polynomial over $\F_2^m$, for which there exists $a \in \F_2^m$,
such that $a \neq c+c'$ for all $c,c' \in \F_2^m$ satisfying $P(c) = P(c') = 0$. Then the error vector 
can be any $e \in \F_2^m$ satisfying $P(e) = 0$. The main result of~\cite{AG} is that 
$\LPSN(n,m)$ can be solved in time $n^{O(m)}$, implying that the linear function can be learnt in polynomial time
when $m$ is constant. 

\subsection{Learning from subset}
We consider here a different model of learning linear functions where the difficulty doesn't come
from the noisy sampling process, but from the fact that instead of obtaining the actual values 
of the sampled elements, we only receive
some partial information about them.

Such a model was first considered
by Friedl et al.~\cite{Friedl} with the {\em Learning From Disequations}
problem $\LFD(q,n)$ where $q$ is a prime number. Here 
we never get sample elements from the kernel of $\ell$, that is we can only sample
$u$ if $\ell(u) \neq 0$, which explains the name of the problem.
Friedl et al.~\cite{Friedl} consider distributions $p$ which are not necessarily uniform
on their support, in fact they only require that $p(u) = p(v)$ whenever
$\ell(u) = \ell(v)$.

The reason to consider this learning problem in~\cite{Friedl} is that the
{\em Hidden Shift} 
problem $\HT(q,n)$, a paradigmatic problem in quantum computing,
can be reduced in quantum polynomial time to $\LFD(q,n)$.
In $\HT(q,n)$ we have oracle access to two injective functions $f_0$ and $f_1$ over $\F_q^n$ with the 
promise that for some element $s \in \F_q^n$, we have $f_1(x) = f_0(x-s)$, for all $x \in \F_q^n$.
The element $s$ is called the {\em hidden shift}, and the task is to find it. It is proven in~\cite{Friedl} 
that $\LFD(q,n)$ can be solved in time $(n+q)^{O(q)}$. This result implies that there exists a quantum algorithm for
$\HT(q,n)$ of similar complexity. When $q$ is constant, these algorithms are therefore polynomial time.

In a subsequent paper~\cite{Ivanyos} Ivanyos extended the work of~\cite{Friedl} 
to the case when $q$ is a prime power, both for 
$\LFD(q,n)$ and $\HT(q,n)$.
The complexity bounds obtained are very similar to the bounds of~\cite{Friedl}, and therefore his results
imply that $\LFD(q,n)$ can be solved in polynomial time, and that $\HT(q,n)$ 
in quantum polynomial time when $q$ is a prime power of constant size.

Observe that the complexity bound $(n+q)^{O(q)}$ is not only not polynomial in $\log q$, but is not even exponential,
in fact it is doubly exponential. Therefore~\cite{Friedl} and~\cite{Ivanyos} not only leave open the question whether,
in general, it is possible to obtain a polynomial time (quantum) algorithm for  $\LFD(q,n)$ and $\HT(q,n)$,
but also the question of the existence of algorithms which are (only) polynomial in $\log q$ and $n$.
These questions are still open today.

In this work we introduce a generalization of the learning problem $\LFD$. 
While in $\LFD$ 
the sampling distribution had to avoid the kernel of $\ell$, in our model the input contains
a set $A \subseteq \F_q$, and we sample from distributions whose support contains only those elements
$u$, for which $\ell(u) \in A$. As in~\cite{Friedl}, we don't require that the sampling distribution is uniform on 
its support, but that the elements with the same $\ell$-value have identical probabilities. We allow these probabilities to be exponentially 
small and even $0$.

\begin{defn}\label{def:dist}
Let $A \subset \F_q$, where $q$ is a prime power, let
$\ell : \F_q^n \rightarrow \F_q$ be a linear function, and let $p$ be a distribution over $\F_q^n$.
We say that the {\em $\ell$-image} of $p$ is $A$ if $\ell(\supp (p)) = A$.
The distribution 
is {\em $\ell$-symmetric} if $\ell(u) = \ell(v)$ implies $p(u) =p(v)$. 
When the $\ell$-image of $p$ is a subset of $A$ and it is also $\ell$-symmetric, we shortly say that it is
an $(A, \ell )$-{\em distribution}. 
\end{defn}

In other words, $p$ is an $(A, \ell )$-distribution if $p$ is constant on each
affine subspace $V_\alpha=\{u\in \F_q^n:\ell(u)=\alpha\}$, for $\alpha \in \F_q$,
and moreover $p$ is zero on $V_\alpha$, whenever $\alpha \not\in A$. 
It is not hard to see that for $|A| < q,$ if $p$ is simultaneously an $(A, \ell )$-distribution
and an $(A, \ell' )$-distribution
then $\ell '$ is a constant multiple of $\ell$. On the other hand, 
non-zero constant multiples of a linear function can not be distinguished in general in this model: for example,
if $A=\F_q \setminus \{0\}$,
then for every $c \neq 0$, an $(A, \ell )$-distribution is also an $(A, c\ell )$-distribution.

\begin{defn}\label{def:LFS}
The {\em Learning From Subset} problem $\LFS(q,n,d)$ is parametrized by 
three positive integers $q, n$ and $d$, where $q$ is a prime power and $2 \leq d \leq q-1$.

\noindent {\em Input:} 
A set $A \subset \F_q$ of cardinality $d$ and a sequence 
of $N$ samples $u_1,\ldots,u_N$ from an $(A, \ell)$-distribution
for some nonzero linear function $\ell : \F_q^n \rightarrow \F_q$.
\\
{\em Output:} A non zero constant multiple of $\ell$. 
\end{defn}

For $d <d'$, an $\LFS(q,n,d)$ instance is also an $\LFS(q,n,d')$ instance, therefore the problem is harder for bigger $d$. 
For $d=1$ the problem is simple because it becomes a system of linear 
equalities which can be solved by Gaussian elimination.
When $d=q$ we don't receive any information from the samples and it is impossible to identify
the linear function.
When $d = q-1$ and $A=\F_q \setminus \{0\}$, the problem $\LFS$ specializes to $\LFD$,
in fact the latter is the hardest instance of the former.

The first main result of our paper is a randomized algorithm for $\LFS(q,n,d)$ whose complexity
depends exponentially on $d$, but only polynomially on $\log q$. This result shows that
the increase of information by reducing the size of the set $A$ can indeed be algorithmically exploited.
More precisely, we show
that  for a sample size $N$ which is a sufficiently large polynomial
of $n^d$, 
there exists a randomized algorithm which in time polynomial in $n^d$ and $\log q$,
with probability 1/2, determines $\ell$ up to a constant factor.

\begin{theorem} \label{thm:LFSconstant} 
There is a randomized algorithm for \LFS$(q,n,d)$ with sample complexity
$(n+d)^{O(d)}$ and running time polynomial in
$\log q$ and $(n+d)^{O(d)}$. 
\end{theorem}

The main interest of this result is that for constant $d$ it gives
a polynomial time algorithm for $\LFS$. For $d = q-1$ and $A=\F_q \setminus \{0\}$ it yields
the same complexity bound as~\cite{Friedl} and~\cite{Ivanyos}. But observe, that even for non constant 
$d = o(q)$, it is asymptotically faster than the algorithms in the above papers.

\subsection{Hidden multiple shifts} 
As we have already said, the original motivation for~\cite{Friedl} to study $\LFD$ was its connection to the hidden
shift problem. This problem was implicitly introduced by Ettinger and H{\o}yer~\cite{EH00}, while
studying the hidden subgroup problem in the dihedral group. The hidden shift problem can be defined in any 
group $G$. We are given two injective functions $f_0$ and $f_1$ mapping $G$ to some arbitrary finite set.
We are promised that for some element $s \in G$, we have $f_1(xs) = f_0(x)$, for every $x \in G$, and the
task is to find $s$. As shown in~\cite{EH00}, when $G$ is abelian, 
the hidden shift in $G$ is quantum polynomial time equivalent to the
hidden subgroup problem in the semidirect product $G \rtimes \Z_2$. 
In the semidirect product the group operation is defined as
$(x_1, b_1).(x_2, b_2) = (x_1 + (-1)^{b_1} x_2, b_1 + b_2)$, and the function $f(x, b) = f_b(x)$ hides the
the subgroup $\{(0, 0), (s, 1)\}.$ The quantum complexity of $\HT$ in the cyclic group $\Z_q$ (or equivalently,
the complexity of the hidden subgroup in the dihedral group $\Z_q \rtimes \Z_2$) is a famous open problem in quantum computing. In~\cite{EH00} there is a quantum algorithm for this problem of polynomial quantum
sampling complexity, but followed by an exponential time 
classical post-processing. The currently best known quantum 
algorithm is due to Kuperberg~\cite{Kup05}, and it is of subexponential complexity
$2^{O(\sqrt{\log q})}$. Note that one could also consider shifts of
non-injective functions. The extension of $\HT$ to such cases can become
quite difficult even over $\Z_2^n$ where $\HT$ for injective functions
is identical to the hidden subgroup problem. 
Results in this direction can be found e.g. in 
\cite{GavinskyRR11}, \cite{ChildsKOR13} and \cite{Roetteler16}.

As one could expect, the polynomial time algorithm for $\LFS$ with constant $d$ 
has further consequences for quantum
computing. Indeed, using this learning algorithm, we can solve in quantum polynomial time some instances
of the hidden multiple shifts problem, a generalization 
of the hidden shift problem, which we define now.

For an element $s\in \Z_q^n$, a subset $H\subseteq \Z_q$ of cardinality at least 2, and a function
$f:\Z_q^n\rightarrow \{0,1\}^l$, where $l$ is an arbitrary positive integer, 
we define the function $f_s:\Z_q^n\times H\mapsto \{0,1\}^l$ as
$f_s(x,h)=f(x-hs).$
We think about $f_s(x,h)$ as the $h$th {\em shift} of $f$ by $s$.
The task in the hidden multiple shift problem is to recover $s$ when we are given oracle access,
for some $f$ and $H$, to $f_s$. This problem doesn't necessarily have a unique solution.
Indeed, let us define $\delta(H,q)$ as the
largest divisor of $q$ such that $h-h'$ is
divisible by $\delta(H,q)$ for every $h,h'\in H$. Pick $h_0\in H$.
Then for any $s'\in \frac{q}{\delta(H,q)}\Z_q^n$
 and $h\in H$, we have
$hs'=h_0s'+(h-h_0)s'=h_0s'$
whence $h(s+s')=hs+h_0s'$ and therefore
$$f_{s+s'}(v,h)=f(v-h(s+s'))=f(v-h_0s'-hs)=f'_s(v,h),$$
where $f'(v)=f(v-h_0s')$. This means that 
$s$ and $s+s'$ are indistinguishable by the set of shifts of $f$,  
and therefore we can only
hope to determine (the coordinates of) $s$ modulo $\frac{q}{\delta(H,q)}$.
When $q$ is a prime number, this problem of course doesn't arise.

\begin{defn}\label{def:HMT}
The {\em Hidden Multiple Shift} problem $\HMS(q,n,r)$ parametrized by
three positive integers $q , n$ and $r$, where  $q >1$ and $2 \leq r \leq q-1$. 

\noindent {\em Input:} A set $H\subseteq \Z_q$ of cardinality $r$.
\\
{\em Oracle input:} 
A function $f_s:\Z_q^n\times H\rightarrow \{0,1\}^l$, where $s\in \Z_q^n$ and 
$f:\Z_q^n \rightarrow \{0,1\}^l$ is an injective function.
\\
{\em Output:} $s \mod  \frac{q}{\delta(H,q)}$. 

\end{defn}

The $\HMS$ problem was first considered by 
Childs and van~Dam~\cite{CD07}.
They investigated the cyclic case $n=1$ and assumed that $H$ is a contiguous
interval and presented a polynomial time quantum algorithm for such an $H$
of size $q^{\Omega(1)}$. Their result could probably be extended to 
constant $n$. However, for ``medium-size" $n$ and $q$, such a result
seems to be very difficult to achieve. Obtaining an efficient algorithm for medium sized $n,q$ is also stated as an open problem \cite{DIKQS14}, 
and it is noted that such a result would greatly simplify their algorithm. 
 Intuitively, for small $H$ the $\HMS$
appears to be ``too close" to the $\HT$ for which the so 
far best result is still what is given in \cite{Friedl}.

For $r=q$, the $\HMS$ problem can be solved in quantum polynomial time.
Indeed, in that case $H=\Z_q$, and $\Z_q^n\times H=\Z_q^{n+1}$ is an abelian group.
The function $f_s$ hides the subgroup generated by $(s,1)$, therefore we have an instance of the
abelian hidden subgroup problem. When $r=1$ the problem is void, there is no hidden shift.
When $r=2$, we have the standard hidden shift problem for which~\cite{Friedl} and~\cite{Ivanyos}
gave a quantum algorithm of complexity $(n+q)^{O(q)} = (n+q)^{O(q+1-r)}$. Their method at a high level is 
a quantum reduction to (several instances of) $\LFS(q,n,q-1)$. These extreme cases suggest a strong connection
between the classical complexity of $\LFS(q,n,d)$ and the quantum complexity of $\HMS(q,n,r)$ when $r = q+1-d$.
Indeed, this turns out to be true. In our second main result we give a polynomial time quantum Turing reduction of 
$\HMS(q,n,r)$ to $\LFS(q,n,q+1-r)$, to obtain an algorithm of complexity $(n+q)^{O((q-r)^{2})}$ for the former problem.

\begin{theorem}  
\label{thm:small-prime}
Let $q$ be a prime. 
Then there is a quantum algorithm
which solves $\HMS(q,n,r)$
with sample complexity and in time $(n+q)^{O((q-r)^{2})}$.
\end{theorem} 

The above Theorem yields a polynomial time algorithm for $\HMS(q,n,r)$ for the case when $q-r$ is constant and $q=n^{O(1)}$. 
We also present a Fourier sampling based algorithm for $\HMS$ which is polynomial time for a different set of parameters satisfying $\frac{r}{q}= 1 - \Omega(\frac{\log n}{n})$. We have the following result.
\begin{theorem} \label{thm:large}
There is a quantum algorithm that solves $\HMS(q,n,r)$ with 
high probability in time $O(\poly(n)(\frac{q}{r}) ^{n+O(1)})$.
\end{theorem}

\subsection{Hidden subgroup problems as instances of $\HMS$}

Certain hidden subgroup problems arise as instances of $\HMS$.
As we have mentioned already, in the
extreme case of $r =q$, we have $\Z_q^n\times \Z_q=\Z_q^{n+1}$
and the function $f_s$ hides the subgroup generated by $(s,1)$.
Thus in this case $\HMS$ is a hidden subgroup problem in
the commutative group $\Z_q^n$. In fact, when $H$ is in a certain
rather strict sense very close to the whole of $\Z_q$, our algorithm
for solving $\HMS(q,n,|H|)$ with input $H$ is (or at least can be considered as)
a version of the standard commutative hidden subgroup algorithm
in $\Z_q^n$. (Note however, that even if only one element of $\Z_q$
is missing from $H$, this method fails as the relative distance
$\frac{1}{q}$ is not small enough if $q$ is much smaller than $n$.)

Another type of the hidden subgroup problems are in certain
generalizations of the affine group considered by Moore et 
al.~\cite{MRRS04}.
Let $\Z_q^*$ denote the multiplicative group of the units of the {\em ring}
$\Z_q$ (that is, the residue classes of integers coprime to $q$).
Let $H$ be a subgroup of $\Z_q^*$. The elements of $H$ act as 
automorphisms of the additive group $\Z_q^n$ by simultaneous
multiplication.
We consider the semidirect product $G=\Z_q^n\rtimes H$ with respect to
this action. Then $G=\Z_q^n\times H$ as a set and the multiplication
in $G$ is given by $(v_1,h_1)*(v_2,h_2)=(h_2v_1+v_2,h_1h_2)$. For $s\in \Z_q^n$
we define
$$K_s=\{(hs-s,h):h\in H\}.$$
Then the $K_s$ are subgroups of $G_H$, in fact these are the
conjugates of the subgroup $K_0=\{(0,h):h\in H\}$. The right cosets
of $K_s$ are the sets $\{(hs+v-s,h):h\in H\}$ where $v$ runs over $\Z_q^n$. 
By taking $w=v-s$, we see that the family of these cosets is the same as
the family $\{(hs+w,h)\}$ where $w$ runs over $\Z_q^n$. Observe that
these are the level sets of a function $f_s$ where $f$ is any injective
function defined on $\Z_q^n$. Thus for multiplicative
subgroup $H\leq \Z_q^*$, the problem
$\HMS(q,n,H)$ is equivalent to the hidden subgroup problem
in $G$ where the hidden subgroup is promised to be a member
of the family $\{K_s:s\in \Z_q^n\}$. If $n$ is even then
$s$ is only determined modulo $\frac{q}{2}$, that is, the function $f_s$
and $f_{s+\frac{q}{2}s'}$ for any injective function $f$ on $\Z_q^n$
have the same level sets for any injective function $f$ on $\Z_q^n$
for any vector $s'\in \Z_q^n$. However, the hidden subgroups
$K_s$ and $K_{s+\frac{q}{2}s'}$ also coincide:
\begin{eqnarray*}
K_{s+\frac{q}{2}s'}&=&\{(h-1)(s+\frac{q}{2}s'),h):h\in H\}=
\{(h-1)s+(h-1)\frac{q}{2}s',h):h\in H\}
\\
&=&
\{(h-1)s,h):h\in H\}=K_s
\end{eqnarray*}
because $(h-1)\frac{q}{2}=0$ for every $h\in H\subseteq \Z_q^*$.

\subsection{Our proof methods} 
The basic idea of the proof of Theorem~\ref{thm:LFSconstant} is
a variant of linearization used in \cite{Friedl} and in \cite{AG},
presented in the flavor of \cite{Ivanyos}. To give a high level 
description, observe that 
every $u$ such that $p(u)\neq 0$ 
is a zero of the polynomial $f_{(A,\ell)}(x)=\prod_{a\in A}(\ell(x)-a).$
By Hilbert's Nullstellensatz, over the algebraic closure of $\F_q$,
the polynomials which vanish on all the zeros of $f_{(A,\ell)}$ are
multiples of $f_{(A,\ell)}$. In particular, every such polynomial
which is also 
of degree at most $d$ must be a scalar multiple of 
$f_{(A,\ell)}$. Interestingly, one could show that this consequence
remains true with high probability, if we replace ``all the zeros" by sufficiently 
many random samples provided that our $(A,\ell)$-distribution is uniform 
(or nearly uniform) in the sense that $p(u)$ (the probability of sampling $u$) is the same (or almost the same) 
for every $u$ such that $\ell(u)\in A$, independently on the actual value of
$\ell(u)$. Therefore, in the (nearly) uniform case one could 
compute a nontrivial scalar multiple of
$f_{(A,\ell)}$ by finding a nontrivial solution of a system
of $N$ homogeneous linear equations in $(n+d)^d$ unknowns (these
are the coefficients of the various monomials in $f_{(A,\ell)}$).
Then $\ell$ could be determined by factoring this polynomial.
This method would be a direct generalization of the algorithms given
in~\cite{Friedl} and~\cite{Ivanyos}. Indeed, in those papers one could
just take $A=\F_q\setminus \{0\}$. However, the proofs (and in case 
of~\cite{Ivanyos} even the algorithmic ingredients) are designed specially
for small $q$ and straightforward extensions would result in algorithms of 
complexity depending exponentially not only on $d$ but on
 $\log q$ as well. Here we give an algorithm that depends polynomially 
on $\log q$ and that works without
any assumption on uniformity. (In the case $A=\F_q\setminus \{0\}$
uniformity can actually be simulated by multiplying the sample
vectors by random nonzero scalars.) Then, instead of divisibility
by $f_{(A,\ell)}$ we prove that, with high probability, the polynomials
that are zero on sufficiently many samples are divisible by
$\ell(x)-a$ for the ``most frequent" value $a\in A$. Then we 
find a scalar multiple of $\ell$ by factoring a nonzero polynomial
from the space of those which are zeros on all the samples. 

The subexponential $\LWE$-algorithm of Arora and Ge \cite{AG} is based
on implicitly solving a problem that can be cast as an instance
of $\LFS$ where one of the coeffecients of the linear function
$\ell$ is known, $0\in A$, and the $(A,\ell)$ distribution is 
such that $0$ is the most likely value.
The problem implicitly used by Arora and Ge in their $\LWE$-algorithm
is the following. Let $A$ be a subset of $\F_q$ of size $d$ containing 
$0$, let $p$ be a probablity distribution on $A$ with $p(0)=\Omega(1/d)$. 
We have access to an oracle that produces pairs $(u_i,b_i)$ such that
$u_i$ are uniformly random vectors from $\F_q^n$ and 
$b_i=\ell(u_i)+e_i$ where $\ell$ is a linear function on $\F_q^n$ and
$e_i$ are chosen from $A$ according to the distribution $p$. The task is 
to determine $\ell$. This problem can be cast as an instance of the
$n+1$-dimensional $\LFS$ as follows. Let $\ell'(x_1,\ldots,x_n,x_{n+1})=
\ell(x_1,\ldots,x_n)-x_{n+1}$ and for $u_i$ let $u_i'$ be the
vector of lenght $n+1$ obtained from $u_i$ by appending $b_i$ as 
the last coordinate. Then $u_i'$ follow an $(A,\ell')$-distribution
so we indeed obtain an instance of $\LFS$. Arora and Ge then prove
a theorem that could be translated into the context of polynomials
as follows. If we have sufficiently many samples $u_i'$ with
$\ell'(u_i')=0$, then, with high probablity, every $n+1$-variable 
polynomial of degree at most $d=|A|$ in which the coefficient 
of $x_{n+1}$ is $-1$ and which vanish on all the $u_i'$s, must
have linear part $\ell'$. (By the linear part of a polynomial
we mean the sum of its monomials of degree one.) It turns out that
that the linear part of every polynomial divisible by the homogeneous linear
polynomial $\ell'$ is a scalar multiple of $\ell'$, so this theorem
would follow easily from our divisibity result.

The algorithm for solving $\HMS(q,n,r)$ in  Theorem~\ref{thm:small-prime} is based on the following. 
After applying some standard preprocessing, we obtain samples of
states that are projections to an $r$-dimensional space
of $\text{QFT}(\ket{(u,s)})$ where $\text{QFT}$ denotes the quantum Fourier transform on $\Z_{q}^{n}$, 
the vector $u \in \Z_{q}^{n}$ is sampled from the uniform distribution on $\Z_{q}^{n}$ and $(\cdot,\cdot)$ denotes 
the standard scalar product of $\Z_q^n$. If we are able to determine the 
scalar product $(u,s)$ for $n$ linearly independent $u$ using the projected states, then $s$ 
can also be computed using Gaussian elimination. However when 
$q$ is not large enough compared to $n$ then the error probability for 
computing $(u,s)$ is too large and we get a system of noisy linear equations 
for which no efficient algorithms are known. Instead, we
can devise a measurement, that at the cost of sacrificing 
a $1-1/q^{O(1)}$ 
fraction of the samples, yields samples $u$ such that
$(u,s)$ belongs to a small subset of $\Z_q$ for sure. More precisely,
the samples follow an $(A,\ell)$ distribution where $A$ is of size
$q-r+1$ and $\ell=(s,\cdot)$. Then we apply 
Theorem~\ref{thm:LFSconstant} and some easy other steps to determine
$s$. So $\LFS$ turns out to be a generalization of the learning
problem $\LFD$ considered in~\cite{Friedl} and~\cite{Ivanyos} 
with applications to quantum computing.

\textbf{Paper organization:} In section \ref{sec:prelim} we state some useful facts and lemmas that are used later for proofs. In section \ref{sec:findlinfct}, we provide the algorithm for $\LFS(q,n,d)$ and prove Theorem \ref{thm:LFSconstant}. In section \ref{sec:large} we propose a Fourier sampling based algorithm for $\HMS(q,n,r)$ and prove Theorem \ref{thm:large}. Finally, in section \ref{sec:HMT} we reduce $\HMS(q,n,r)$ to $\LFS(q,n,q-r+1)$ and prove Theorem \ref{thm:small-prime}.

\section{Preliminary facts}  \label{sec:prelim} 
\subsection{Algebraic preliminaries} 

We collect here some 
well known facts and lemmas 
that will be useful for proofs in later sections. We recall first the well known 
Hadamard inequality and the expression for the determinant of a Vandermonde matrix.

\begin{fact} \label{fact:hadamard} 
\cite{H93} Let $M \in \C^{n\times n}$ be a matrix with column vectors $v_{i} \in \C^{n}$, then $|det(M)| \leq \prod_{i \in [n]}  \norm{ v_{i}} $. 
\end{fact} 
\begin{fact}
\label{fact:Vandermonde-det}
Let $x_1,\ldots,x_m$ be variables. Then the following identity
holds in the polynomial ring $\Z[x_1,\ldots,x_m]$.
\begin{equation*} 
\det\begin{pmatrix}
1 & x_1 & x_1^2 & \cdots & x_1^{m-1} \\
1 & x_2 & x_1^2  &\cdots & x_2^{m-1} \\
\vdots & \vdots &\vdots &\ddots & \vdots \\
1 & x_m & x_m^2  &\cdots & x_m^{m-1} 
\end{pmatrix}=\prod_{1\leq j<i\leq m}(x_i-x_j).
\end{equation*} 
\end{fact}
We next state the Schwartz-Zippel lemma \cite{Z79, S80} and then prove a variant that is required for the proof of Theorem \ref{thm:LFSconstant}. 
\begin{fact} \label{fact:szippel} 
\cite{Z79, S80}  Let $g\in \F_{q}[x_{1}, x_{2}, \cdots, x_{n}]$ be a non zero $n$-variate polynomial of degree $d$ over a finite field $\F_{q}$ for $n,d\geq 1$. Let $S$ be a subset of $\F_{q}$
and let $u$ be sampled from the uniform distribution on $S^{n}$, then,  
$$\Pr_{{\underline a} \sim S^{n}} [ g(u_{1}, u_{2}, \cdots, u_{n}) = 0 ] \leq \frac{d}{ |S|}.$$ 
\end{fact} 
\noindent We require a variant of the Schwartz-Zippel lemma where the polynomial $g(x)$ is not divisible by 
a linear function $\ell(x)$ and the samples are drawn from an affine
 subspace $V_\alpha=\{u\in \Z_q^n:\ell(u)=\alpha\}$
for a fixed $\alpha \in \F_{q}$. 
\begin{lemma}
\label{lem:poly-nz}
Let $g(x_1\ldots,x_n)$,
be a degree $d$ polynomial in $\F_q[x_1,\ldots,x_n]$ that is not divisible
by $\ell(x_1,\ldots, x_n)-\alpha$ where $\alpha\in \F_q$ and $\ell(x_1,\ldots,x_n)$ is a nonzero
homogeneous linear polynomial. Let 
$u=(\beta_1,\ldots,\beta_n)$ be sampled uniformly at random
from the affine subspace $V_\alpha=\{u\in \Z_q^n:\ell(u)=\alpha\}$, then 
$$
\Pr_{u \sim V_{\alpha}}  [ g(u) =  0 ] \leq \frac{d}{q}. 
$$
\end{lemma}

\begin{proof}
Let $v_1,\ldots,v_n$
be $n$ linearly independent vectors from $\F_q^n$ such that
$\ell(v_1)=1$ and $\ell(v_i)=0$ for $i=2,\ldots,n$. If $y= \sum_{i \in [n]} v_{i} x_{i}$ 
then we have that $\ell(y)=x_{1}$. Thus, in the basis spanned by the $v_{i}$s, the vector $u=(\alpha, \beta_2,\ldots,\beta_n)^t$ where 
$(\beta_2, \beta_3, \cdots, \beta_n)$ is a uniformly random vector 
from $\F_q^{n-1}$ is a sample from the uniform distribution on $V_\alpha$.

Let $h(y_1,\ldots,y_n)$ be be the quotient and
$r(y_1,y_2,\ldots,y_n)$ be the remainder when $g(y_1,\ldots y_n)$
is divided by the polynomial $y_1-\alpha$
in the univariate polynomial ring $\F_q(y_2,\ldots,y_n)[y_1]$, that is 
$$g(y_1,\ldots,y_n)=(y_1-\alpha)h(y_1,\ldots,y_n)+r(y_2,\ldots,y_n).$$

\noindent Clearly, $r(y_1,\ldots,y_n)$ has degree
zero in $y_1$ and is a polynomial of degree at most 
$d$ in the variables $y_{2}, y_{3}, \cdots, y_{n}$. 


Substituting the uniformly random vector $u=(\alpha, \beta_2,\ldots,\beta_n)^t$ for $y$ in the above relation, the term
$(u_1-\alpha)h(u_1,\ldots,u_n)=0$ while, by the Schwartz-Zippel Lemma 
(Fact~\ref{fact:szippel}), we have
that $r(\beta_2,\ldots,\beta_n)=0$ with probability at
most $d/q$ if $r$ is a non zero polynomial. Thus, for polynomial $g$ that does not have $\ell(x_1,\ldots, x_n)-\alpha$
as a factor (so that $r(y_2,\ldots,y_n) \neq 0$ in above equation), we have that $\Pr_{u \sim V_{\alpha}} [ g(u)=0 ] \leq d/q$. 
\end{proof}

\section{An algorithm for LFS}
\label{sec:findlinfct}

In this section we prove Theorem~\ref{thm:LFSconstant}. 
Let $p$ be an $(A, \ell)$-distribution on $\F_q^n$, where $|A| = d$. 
We define $\alpha_p$ as the element $\alpha \in A$ for which 
$\Pr[ \ell(u) = \alpha]$ is maximal (breaking a tie arbitrarily).
We start the proof with our main technical lemma 
which links $p$ to the  space of $n$-variable polynomials of degree $d$.

\begin{lemma}
\label{lem:poly-shrink}
Let 
$N=\Omega\left(\binom{n+d}{d} {d^2\log\binom{n+d}{d}}
\right)$
and let $u_1,\ldots,u_N$ be  sampled independently  from  
an $(A,\ell)$-distribution on $\F_q^n$, where $|A| = d<q$. 
Then with probability at least $1/2$, every
polynomial $g(x_1,\ldots, x_n)$ 
over $\F_q^n$ of degree at most $d$, for which $g(u_i)=0$ for $i=1,\ldots,N$,
is divisible by $\ell(x_1,\ldots, x_n)-\alpha_p$.
\end{lemma}

\begin{proof}
For $j=0,\ldots,N$ we set $P_j$ to be the set of polynomials in
$\F_q[x_1,\ldots,x_n]$ of degree at most $d$ which take zero 
value on the first $j$ samples: 
$$P_j =\{g(x_1,\ldots,x_n):\deg g\leq d \mbox{~and~}g(u_i)=0 
\mbox{~for~}i=1,\ldots,j\}.$$
In particular, $P_0$ is the set of all polynomials
of degree at most $d$. 
We consider $P_0$ as a vector space of dimension $\binom{n+d}{d}$
over $\F_q$. 
Since, for $u\in \F_q^n$, the map $g\mapsto g(u)$ is linear
on $\F_q[x_1,\ldots,x_n]$, we conclude that $P_0,\ldots,P_N$
is a non-increasing sequence of subspaces of $P_0$. 

Set $\pi=\Pr[ \ell(u) = \alpha_p ]$,
and observe that $\pi \geq \frac{1}{d}$.
Let $P'$ be the set of polynomials from $P_0$
which are divisible by $\ell(x_1,\ldots,x_n)-\alpha_p$.
Then an equivalent way to state the lemma 
is that 
$P_{N} \subset P'$, with probability at least $1/2$.

We first claim that, for every $j=1,\ldots,N$, 
\begin{equation}
\label{conditional}
\Pr[P_{j}=P_{j-1} | P_{j-1}\not \subseteq P'] \leq 1- \frac{1}{d(d+1)}.  
\end{equation}

In order to prove this bound, we note that the condition $P_{j-1}\not \subseteq P'$ means that
there exists a non zero $g \in P_{j-1} \setminus P'$. Fix such a $g$. 
The event 
$P_j=P_{j-1}$ is equivalent to
$f(u_j)=0$, for all $f\in P_{j-1}$. 
Therefore
\begin{align*}
\Pr[P_{j}=P_{j-1} | P_{j-1}\not \subseteq P']  &  \leq \Pr[ \forall f \in P_{j-1}, ~ f(u_j) = 0] \\
&  \leq \Pr [ g(u_j) = 0].
\end{align*}
The probability that $g(u_j)=0$ can be bounded as follows:
\begin{align*} 
\Pr [ g(u_{j})=0 ] &\leq \Pr [ g(u_{j})=0 | \ell(u_j)\neq \alpha_p] \cdot (1-\pi) + \Pr [ g(u_{j})=0 | \ell(u_j)= \alpha_p]
\cdot \pi \\
& \leq (1-\pi ) + \pi \frac{d}{q}. 
\end{align*} 
The first inequality follows simply by decomposing the event $g(u_j) = 0$ according to whether 
$\ell (u_j)$ is different from, or equal to $\alpha_p$.
In the second case, which happens with probability $\pi$, Lemma~\ref{lem:poly-nz} is applicable
and it states that $g(u_{j})=0$ with probability at most $d/q$. This explains the second inequality.
Using $\pi \geq 1/d$ and $q>d$, 
a simple calculation 
gives 
$$
1-\pi+\pi\frac{d}{q} \leq 1-\frac{1}{d(d+1)},
$$
from which the inequality~(\ref{conditional}) follows.


We can use the conditional probability in (\ref{conditional}) to
upper bound the probability of the event that
$P_{j-1}\not\subseteq P'$ and 
$P_{j}=P_{j-1}$ hold simultaneously. But if $P_j=P_{j-1}$
then $P_{j-1}\subseteq P'$ is equivalent to
$P_j\subseteq P'$, therefore we can infer, for every $j=1,\ldots,N$,
\begin{equation*}
\label{conjunction}
\Pr[P_j\not\subseteq P'\mbox{~and~}P_j=P_{j-1}] \leq 1-\frac{1}{d(d+1)}.
\end{equation*}

Iterating the above argument $k$-times, we obtain,
for every $k \leq N$ and $j \leq N-k+1$, 
\begin{equation}
\label{iterated}
\Pr[P_{j+k -1} \not\subseteq P'\mbox{~and~}P_{j+k-1}=P_{j-1}] 
\leq \left (1-\frac{1}{d(d+1)}\right )^k.
\end{equation}

Indeed, as before, we can bound the probability on the left hand side by the conditional probability
$\Pr[ P_{j+k-1}=P_{j-1} | P_{j+k -1} \not\subseteq P'].$ Under the condition $P_{j+k -1} \not\subseteq P'$,
there exists a non zero $g \in P_{j+k-1} \setminus P'$, and we fix such a $g$. Then
\begin{align*}
\Pr[P_{j+k -1} \not\subseteq P'\mbox{~and~}P_{j+k-1} ]
& \leq \Pr[g(u_{j+i}) = 0, \mbox{~for~} i = 0, \ldots, k-1] \\
& \leq \prod_{i=0}^{k-1}  \Pr[g(u_{j+i}) = 0] \\
& \leq (1-\frac{1}{d(d+1)})^k,
\end{align*}
where for the second inequality we used that the samples $u_{j+i}$ are independent. 

Taking $k=\Omega(d^2\log\binom{n+d}{d})$,
$N = (\binom{n+d}{d} +1)k$ and $j = mk +1$, for $m = 0, 1, \ldots , \binom{n+d}{d} $, in
inequality (\ref{iterated}), we get 

\begin{equation*}
\label{bigstep}
\Pr[P_{(m+1)k} \not\subseteq P'\mbox{~and~}P_{(m+1)k}=P_{mk}] 
\leq \frac{1}{2{\binom{n+d}{d}}^{-1}}.
\end{equation*}

For the complement of the union of these $\binom{n+d}{d} +1$ events, we derive then 
\begin{equation*}
\Pr[\bigcap_{m=0}^{\binom{n+d}{d} } 
\Big( P_{(m+1)k} \subseteq P'\mbox{~or~}P_{(m+1)k} \subset P_{mk}
\Big)] \geq \frac{1}{2}.
\end{equation*}
If $P_{(m+1)k} \subset P_{mk}$ for some $m$, then $\dim (P_{(m+1)k} ) < \dim (P_{mk})$.
We can not have simultaneously $\dim (P_{(m+1)k} ) < \dim (P_{mk})$, for 
$m = 0, 1, \ldots , \binom{n+d}{d} $,  because
otherwise $\dim (P_N)$ would be negative. Therefore,
with probability at least 1/2,  $P_{(m+1)k} \subseteq P'$, for some 
$m \leq \binom{n+d}{d}$, implying $P_{N} \subseteq P'$.
\end{proof}

We now present an algorithm 
for $\LFS(q,n,d)$ and 
show that it solves the problem efficiently when the input contains a polynomially large number of 
samples $u_1,\ldots,u_N \in \F_q^n$  
from an $(A,\ell)$-distribution, with $|A| = d$ constant.


\begin{algorithm}[H]
\caption{Algorithm for \LFS$(q,n,d)$.} \label{alg:LFS}
\begin{algorithmic}[1]
\REQUIRE 
A set $A \subset \F_q$ of cardinality $d$ and a sequence 
of $N$ elements $u_1,\ldots,u_N$ from $\F_q^n$.
\\
\begin{enumerate} 
\item Find a nonzero polynomial $g(x_1,\ldots,x_n)$ of degree at most
$d$ over $\F_q$, if exists, such that $g(u_i)=0$ for $i=1,\ldots,N$. 
 \\
\item Compute the  linear factors of $g$. 
\\
\item Find a linear factor $f$ of $g$ and a 
nonzero element $\gamma\in \F_q$, if exist, such that 
$\gamma(f(u_i)-f(0))\in A$, for $i=1,\ldots,N$.
Return the linear function $\gamma(f(x_1,\ldots,x_n)-f(0))$. \\
 \end{enumerate} 
\end{algorithmic}
\end{algorithm}

\begin{theorem} \label{thm:LFS} 
There is a randomized implementation of {\bf Algorithm~\ref{alg:LFS}}
which runs in time polynomial in $\log q,\binom{n+d}{d}$ and $N$.
Moreover, when $u_1,\ldots,u_N$ are independent samples  from  
an $(A,\ell)$-distribution on $\F_q^n$ where $|A| = d$ and 
$N=\Omega\left(\binom{n+d}{d} {d^2\log\binom{n+d}{d}}\right)$,  then it 
finds successfully $\ell$ up to a constant factor with probability at least $1/2$. 
\end{theorem}


\begin{proof}
We first describe the randomized implementation with the claimed running time.
Throughout the proof by polynomial time we mean time
polynomial in $\log q,\binom{n+d}{d}$ and $N$. 
For Step 1, we consider the $\binom{n+d}{d}$ dimensional vector space of $n$-variable polynomials over $\F_q$
of degree at most $d$. 
The system of requirements $g(u_i)=0$, for $i=1,\ldots,N$, is
equivalent to a system of $N$ homogeneous linear equations 
for the $\binom{n+d}{d}$ coefficients of $g$,
where in the $i$th equation,
the coefficients of the variables are the values of the monomials taken at $u_i$.
Therefore 
a solution, if exists, can be computed in polynomial time using standard linear algebra. 

We use Kaltofen's
algorithm~\cite{Kaltofen85} 
to find the irreducible factors of $g$. The finite field case of Kaltofen's
algorithm is discussed in section 4.3 in \cite{GK83}. The algorithm is a Las Vegas
randomized algorithm that runs in polynomial time 
given the representation of the input polynomial
as a list of all coefficients. We can then easily select the 
linear factors out of the irreducible factors, therefore Step 2 can also be done in polynomial time.

For  Step 3, note that $g$ has at most $d \leq n$ linear factors, therefore it is enough to
see that each individual factor $f$ can be dealt with in polynomial time.
This can be done as follows. If $f(u_i)=f(0)$ for every $i$, then an appropriate $\gamma$ can be found
if and only if $0\in A$. Indeed, if $0\in A$ then any 
nonzero $\gamma$ satisfies the condition, while otherwise no satisfying $\gamma$ exists.
Otherwise, pick any $i$ such that
$\beta = f(u_i)-f(0)\neq 0$ and try $\gamma=\alpha/\beta$ for every
$\alpha\in A$. 

We now turn to the proof of correctness of the algorithm when the samples 
come from an $(A,\ell)$-distribution.
As 
$\prod_{\alpha\in A}\left(\ell(u_i)-\alpha\right)=0$,
for every $i$,
the algorithm finds a nonzero polynomial $g$ in Step 1.
By Lemma~\ref{lem:poly-shrink}, with probability at least 1/2, every polynomial 
of degree at most $d$, which is zero on $u_i$, for $i=1,\ldots,N$,
is divisible by $\ell(x_1,\ldots,x_n)-\alpha_p$. Assume that this is the case.
Then, in particular, $g$ has a linear factor $f(x)$ which is a constant multiple of $\ell(x)-\alpha_p$,
that is $f(x) = \beta (\ell(x) - \alpha_p)$, for some non zero $\beta \in \F_q$.
It is easy to check that for $\gamma = \beta^{-1}$, we have 
$$ \gamma ( f(x) - f(0)) = \ell(x),$$
and therefore $ \gamma ( f(u_j) - f(0)) \in A$, for $i = 1, \ldots , N$.
Thus the algorithm in its last step will find successfully and return
a linear function $\ell'(x)$ such that $\ell'(u_i)\in A$, for every $i$. 

To finish the proof, we claim that $\ell'(x)$ is a constant multiple of $\ell(x)$.
The polynomial
$h(x_1,\ldots,x_n)=\prod_{\alpha \in A}(\ell'(x_1,\ldots,x_n)-\alpha)$
is zero on every $u_i$ and hence, by our assumption, 
$h(x_1,\ldots,x_n)$ is divisible by $\ell(x_1,\ldots,x_n)-\alpha_p$.
Then, as $\F_q[x_1,\ldots,x_n]$ is a unique factorization domain,
there exists $\alpha \in A$ such that 
$\ell'(x_1,\ldots,x_n)- \alpha$ is a scalar multiple of
$\ell(x_1,\ldots,x_n)-\alpha_p$, 
implying the claim.
\end{proof}
Theorem~\ref{thm:LFSconstant} is an immediate consequence of this result. For constant $d$ we have the
following corollary.

\begin{corollary}
There is a randomized algorithm that solves $\LFS(q,n,d)$ for constant $d$ with sample complexity $\poly(n)$ and 
running time $\poly(n, \log q)$. 
\end{corollary}
We next present our algorithms for $\HMS(q,n,r)$, we first give a basic Fourier sampling based algorithm in section \ref{sec:large} and then 
an algorithm that reduces $\HMS(q,n,r)$ to $\LFS(q,n, q-r +1)$ in section \ref{sec:HMT}.

\section{Fourier sampling algorithm for $\HMS(q,n,r)$} 
\label{sec:large}
We describe first the standard pre-processing steps for the $\HMS(q,n,r)$ problem, for our algorithms we assume that 
the inputs have been pre-processed in this manner. 

A quantum algorithm for the $\HMS(q,n,r)$ problem is given oracle access to
$f_s:\Z_q^n\times H\rightarrow \{0,1\}^l$. 
We describe the standard pre-processing procedure applicable in this context.
We start with the uniform superposition, append a register consisting of $l$ qubits, 
initialized to $0$ and query the oracle for $f_{s}$ to obtain, 
$$\frac{1}{\sqrt{q^n r}}\sum_{v\in \Z_q^n}\sum_{h\in H}\ket{v}\ket{h} \to \frac{1}{\sqrt{q^n r}}\sum_{v\in \Z_q^n}\sum_{h\in
H}\ket{v}\ket{h}\ket{f_{s}(v, h)}.$$
The last $l$ qubits are then measured to obtain the state, 
$$\psi_s^w:=\frac{1}{\sqrt{r}}\sum_{h\in H}\ket{w+hs}\ket{h},$$
where $w\in \Z_q^n$ is uniformly random. This $w$ is the 
unique element of $\Z_q^n$ such that the measured value for
the function $f_{s}$ equals $f(w)$. 

Let 
$\omega= e^{2\pi i /q}$ be the $q$-th root of unity and let $(u,v)$ be the 
standard inner product for vectors $u, v \in \Z_q^{n}$, that is 
$(u,v) = \sum_{i \in [n]} u_{i} v_{i} \mod q$. It is usual to apply the 
Fourier transform on $\Z_q^n$ to the states $\psi_s^w$ to get
$$\frac{1}{\sqrt{q^n r}}
\sum_{u\in \Z_q^n}\sum_{h\in H}\omega^{(u,w+hs)}
\ket{u}\ket{h}=
\frac{1}{\sqrt{q^n}}\ket{u}\otimes \omega^{(u,w)}\phi_s^u,$$
where
$$
\phi_s^u:=\frac{1}{\sqrt{r}} \sum_{h\in H}\omega^{(u,hs)}\ket{h}. 
$$

If we measure $u$, we obtain a pair $(u,\tau)$,
where $u$ is a uniformly random element of $\Z_q^n$ and
$\tau=\phi_s^u$. It is standard to pre-process the inputs to the 
$\HMS$ in this manner. We can therefore assume without loss of generality that the inputs for $\HMS(q,n,H)$
are $N$ samples of the form $(u,\phi_s^u)$ for uniformly random  $u\in \Z_q^n$, and the 
goal of the algorithm is to recover the secret $s \in \Z_q^{n}$.

We next propose a Fourier sampling based algorithm for $\HMS(q,n,r)$ and prove Theorem \ref{thm:large}. 
The basic idea for the algorithm is to consider the input state $\phi_s^u=\frac{1}{\sqrt{r}}\sum_{h\in H}\omega^{(u,hs)}\ket{h}$ for $\HMS(q,n,r)$, 
 as an approximation to the state 
$$\kappa_s^u: =\frac{1}{\sqrt{q}}\sum_{h=0}^{q-1}\omega^{(u,hs)}\ket{h}.$$
The inner product between the two states is
${\phi_s^u}^{\dagger}\cdot\kappa_s^u  =
 \frac{1}{\sqrt{qr}} \sum_{h\in H} 1 =\sqrt{r/q}.$
 
The inverse Fourier transform on $\Z_q$, when applied to $\kappa_s^u$ 
gives $\ket{(u,s)}$. If we could determine the inner products $\ket{u,s}$ for 
a set of $n$ linearly independent $u_i$ for prime $q$, 
then $s$ can be determined by 
solving a system of linear equations. More generally, in order to make this approach work $k$ should be large enough so that the $u_{i}$ generate $\Z_q^{n}$, in this case the secret $s$ can be recovered from the inner products using linear algebra. In fact, the following Lemma from \cite{Pomerance} shows that the additive group $\Z_q^n$ is generated by $k=n + O(1)$ random elements of $\Z_q^n$ with constant probability.   
  \begin{fact} \label{fact:generate} \cite{Pomerance} 
Let $G$ be a finite abelian group with a minimal generating set of size $r$. The expected number of elements chosen independently and uniformly at random from $G$ 
such that the chosen elements generate $G$ is at most $r+ \sigma$ where $\sigma < 2.12$ is an explicit constant. 
\end{fact} 
The above 
fact holds for any abelian group, for the special case of $\Z_{q}^{n}$ we have $r=n$ and the constant $\sigma$ can be taken to be $1$ \cite{vincenzo}. 
We therefore have that $k=2n+ O(1)$ random elements of $\Z_q^n$ generate the additive group $\Z_q^n$ with constant probability. 

If we apply the Fourier transform to each $\phi_s^{u_i}$, with probability $(r/q)^{k/2}$ we obtain the scalar products of $s$ with the members of a generating set for $\Z_q^n$.
The answer $s$ may be verified by repeating the experiment for 
$\poly(n) (q/r)^{k/2}$ trials and finding the most frequently occurring solutions over the different trials.

\begin{thm:large}
There is a quantum algorithm that solves $\HMS(q,n,r)$ with high probability in time $O(\poly(n)(\frac{q}{r}) ^{n+ O(1)})$.
\end{thm:large}
\noindent We next show that the above algorithm runs in time $\poly(n)$ 
for parameters $q,r$ such that $\frac{r}{q} = 1 - \Omega(\frac{\log n}{n})$. For this choice of parameters, we 
can bound the factor $(\frac{q}{r}) ^{n + O(1)}$ in the running time bound above as follows, 
$$
 \left (  \frac{r}{q} \right ) ^{n  + O(1) } \geq \left ( 1 - \frac{ c_1 \log n }{n} \right )^{c_2 n + c_{3}} \geq  e^{-c \log n} = n^{-O(1)} 
$$
where $c, c_{1}, c_{2}$ are suitable constants. We therefore have,

\begin{corollary}
\label{thm:large-cor}
If $\frac{r}{q} = 1 - \Omega(\frac{\log n}{n})$, then there is a quantum algorithm that solves $\HMS(q,n,r)$ with high probability in time $\poly(n)$.
\end{corollary}

\section{Reducing $\HMS(q,n,r)$ to $\LFS(q,n, q-r +1)$}\label{sec:HMT} 

In this section we assume that $q$ is a prime number and work over the field $\F_{q}$. 
Recall that the input for $\HMS(q,n,r)$ is a collection of samples of vector-state pairs $(u, \phi_s^u)$ where $u$ 
is a uniformly random vector from $\F_q^n$, and 
$\phi_s^u=\frac{1}{\sqrt r}\sum_{h\in H}\omega^{(u,s)h}\ket{h}.$ For $t\in \F_q$ define the state
$$\mu_t:=\frac{1}{\sqrt r}\sum_{h\in H}\omega^{ht}\ket{h},$$
so that $\phi_s^u=\mu_{(u,s)}$. 

The approach in Section~\ref{sec:large}  recovers the inner product $(u_{i}, s)$ for $O(n)$ random vectors $u_{i}$ and then 
uses Gaussian elimination to determine $s$ with high probability. However, the $\mu_t$'s are only nearly orthogonal to each other, so 
the measurement in Section~\ref{sec:large} may fail to recover the correct value of $(u,s)$ with
probability too large for our purposes. 

A particularly interesting case is when 
$q=\poly(n)$ and $c=q-r$ is a constant for which we provide a polynomial time algorithm in 
Corollary~\ref{cor:eff}. In this case, the error probability for 
the measurement in Section~\ref{sec:large} is $1-r/q=c/q$, that is there are a constant expected number of errors 
for every $q$ samples. If $q=O(n^{\alpha})$ for $\alpha <1$ then there 
are $O(n^{1-\alpha})$ errors in expectation for every $n$ samples. There are no known polynomial time algorithms 
for recovering the secret $s \in \Z_q^{n}$ from a system of $n$ linear equations where 
an $O(n^{1-\alpha})$ fraction of the equations are incorrect for a constant $\alpha$.

\subsection{The reduction} 

Instead, we reduce $\HMS(q,n,r)$ to $\LFS(q, n, d)$ with $d=(q-r)+1$,
$A=\{r-1,\ldots,q-1\}$ and the linear function $\ell(\cdot)$ given by $\ell(x)=(s,x)$. 
We then use the Algorithm \ref{alg:LFS} to recover a scalar 
multiple of $s_{0}=\lambda s$. Further, we show that the scalar $\lambda$ can be recovered 
efficiently.

The reduction performs a quantum measurement on $\phi_s^u$  to determine if $(u,s)$ belongs to $A=\{r-1,\ldots,q-1\}$. 
We discard the $u$'s which do not belong to $A$, and also some of the $u$'s such that 
$(u,s)\in A$ to obtain samples from an $(A, \ell)$ distribution. 
We next provide a sketch 
of the reduction $\HMS(q,n,r)$ to $\LFS(q, n, d)$, the reduction is analyzed over the 
next few subsections and a more precise statement is given in Proposition \ref{prop:sampling}. 

Let $V$ be the hyperplane spanned by $\mu_0,\ldots,\mu_{r-2}$. Let 
$(u, \phi_s^u)$ be a pair from the input samples. We perform
the measurement on $\phi_s^u$ according to
the decomposition of $\C^{r}=V\oplus V^\perp$, and retain 
$u$ if and only if the result of the measurement is ``in $V^\perp$". Otherwise we discard $u$. 
An efficient implementation of the measurement in $(V, V^{\perp})$ 
is given in subsection \ref{ssec:mmt}.

Observe that measuring a state $\mu_j$ ``in $V^\perp$"
is only possible if $\mu_j\not\in V$, in particular
$j\not\in \{0,\ldots,r-2\}$. Thus if we measure 
$\phi_s(u)$ ``in $V^\perp$" we can be sure that 
$(s,u)$ is in $A=\{r-1,\ldots,q-1\}$. We only keep $u$ from a sample pair $(u,\tau)$ if this measurement,
applied to the state $\tau$, results ``in $V^\perp$". The $u$'s that are retained are 
samples from an $(A, \ell)$ distribution over $\F_q^n$.

\subsection{The success probability}

We bound the probability of retaining a sample pair $(u, \phi_s^u)$ for this procedure. 
We bound the success probability for the special case when $(s,u)=r-1$. As $u$ is uniformly random over $\F_q^{n}$ the value of $(s,u)$ 
is uniformly distributed over $\Z_q$, this bound therefore suffices for our purposes.

In the standard basis $\ket{h}$ ($h\in H$), the vector $\mu_t$ has entry $\frac{1}{\sqrt r}\omega^{ht}$ 
in the $h$-th position. Let $A \in \C^{r \times r}$ be the matrix with rows from 
the collection $\{\mu_t:0\leq t\leq r-1\}$, that is
\begin{equation} \label{amat} 
A=\frac{1}{\sqrt{r}}
\begin{pmatrix}
1 & \omega^{h_1} & \cdots & \omega^{h_1(r-1)} \\
1 & \omega^{h_2} & \cdots & \omega^{h_2(r-1)} \\
\vdots & \vdots &\ddots & \vdots \\
1 & \omega^{h_r} & \cdots & \omega^{h_r(r-1)} \\
\end{pmatrix},
\end{equation} 
where $h_1,\ldots,h_r$ are the elements of $H$, say, in increasing order. 
The matrix $A$ is
$\frac{1}{\sqrt{r}}$ times a Vandermonde matrix and as such,
by Fact~\ref{fact:Vandermonde-det},
has determinant
$$r^{-r/2}\prod_{j<i\leq r}(\omega^{h_i}-\omega^{h_j}).$$
In particular, the states $\mu_0,\ldots,\mu_{r-1}$ are 
linearly independent. With a more careful analysis, we
show in Lemma \ref{lem:Vandermonde-bound} below that $\det A$ is sufficiently far from zero.

\medskip 
\begin{lemma} \label{lem:Vandermonde-bound} 
Let $c=q-r$ and $A$ be the matrix in  \eqref{amat} then $ |\det A^*A|=\Omega(q^{-c^{2}} \left ( \frac{q}{r} \right )^{r}).$
\end{lemma} 

\begin{proof} 
As $|\det A^*A|= |\det A|^{2}$ we have
\begin{eqnarray*}
|\det A^*A|&=&
r^{-r}\left|\prod_{i\in H}\prod_{j\in H\setminus\{i\}}(\omega^{i}-\omega^{j})\right|\\
&=&r^{-r}\left|\left(
\prod_{i\in H}\prod_{j\in \Z_q \setminus\{i\}}(\omega^{i}-\omega^{j})
\right)\left(
{\prod_{i\in H} \prod_{j\in\overline H}(\omega^{i}-\omega^{j})}\right)^{-1}
\right| \\
&=&r^{-r}\left|\left(
\prod_{i\in H}\prod_{j\in \Z_q\setminus\{i\}}(\omega^{i}-\omega^{j})
\right)\left(
\prod_{j\in\overline H}\prod_{i\in \Z_q \setminus\{j\}}(\omega^{i}-\omega^{j})
\right)^{-1} \right.\\
& & \mbox{~~~~}\left.\cdot\left(
\prod_{j\in \overline H}\prod_{i\in {\overline H}\setminus\{j\}}(\omega^{i}-\omega^{j})
\right)\right| \\
&=&r^{-r}q^rq^{-c}
\prod_{j\in \overline H}\prod_{i\in {\overline H}\setminus\{j\}}|\omega^{i}-\omega^{j}| \\
&\geq & \left(\frac{q}{r}\right)^r q^{-c} 
\left|1-e^{\frac{2\pi i}{q}}\right|^{c(c-1)} 
= \Omega\left ( \left(\frac{q}{r}\right)^r q^{-c^{2}} \right).
\end{eqnarray*}
We used the identity $\prod_{j\in \Z_q\setminus \{i\}}(\omega^i-\omega^j)=
\omega^i\prod_{j\in \Z_q\setminus \{0\}}(1-\omega^j)=q\omega^i$
for the fourth equality. The final inequality follows as 
$|\omega^i-\omega^j|\geq |1-e^{\frac{2\pi i}{q}}| = \Omega(\frac{1}{q})$. 

\end{proof} 

\noindent Using the above lemma, we bound the probability of retaining $u$ if $(u,s)=r-1$. It might be possible to prove similar bounds for other values of $(s,u) \in A$, 
however the bound for the particular value $r-1$ suffices for our purpose.

\begin{lemma}
\label{lem:measure-bound}
The $(V,V^\perp)$-measurement applied to a state of the
form $\mu_t$, returns ``in $V$" with probability $1$ if
$t\in \{0,\ldots,r-2\}$, while for $t\in \{r-1,\ldots,q-1\}$,
the probability that ``in $V^\perp"$ is returned depends only
on $t$ and is $\Omega\left (q^{-c^{2}}\left( \frac{q}{r} \right )^{r} \right)$ for $t=r-1$.
\end{lemma}

\begin {proof}
If $t\in \{0,\ldots,r-2\}$, then $\mu_t \in V$ and ``in $V$" is obtained with probability $1$. 
If $t\in \{r-1,\ldots,q-1\}$, the probability of obtaining "in $V^\perp"$ 
is the square of the length of the component of $\mu_{t}$ orthogonal
to $V$ and thus depends only on $t$. 

Let
$\mu$ be the component of $\mu_{r-1}$ orthogonal to
$V$. The probability that $\mu_{r-1}$ is measured
``in $V^\perp$" is then  $|\mu|^{2}$. 
Let $A_0$ be the left $r\times (r-1)$ sub-matrix of $A$
in \eqref{amat}.
Then 
$$
|\mu|^2=\frac{|\det A^*A|}{ |\det A_0^*A_0|} \geq |\det A^*A|= 
\Omega\left (q^{-c^{2}} \left (\frac{q}{r} \right )^{r}\right ).$$
Note that we used Hadamard's inequality (Fact~\ref{fact:hadamard}) 
for the bound $|\det A_0^*A_0| \leq 1$
and Lemma \ref{lem:Vandermonde-bound} for estimating $|\det A^*A|$.
\end{proof}

\medskip

\subsection{Implementation of the measurement.} \label{ssec:mmt}
The $(V, V^{\perp})$ measurement acts on $O(\log q)$ qubits. Using universality constructions and 
the Solovay-Kitaev theorem, it is well known that an arbitrary unitary operator can be approximated 
using an exponential number of elementary gates in the number of qubits. 
\begin{fact} \label{fact:approx_unit} 
\cite{NC02} An arbitrary unitary operation $U$ on $t$ qubits can be simulated to error $\epsilon$ using
$O( t^{2}  4^{t} \log^{c} (t^{2}4^{t}/\epsilon))$ elementary gates.  
\end{fact} 
\noindent The ability to implement an arbitrary unitary operation on $\log q$ qubits implies the ability to perform the measurement $(W, W^\perp)$
for an arbitrary subspace $W \subset \C^{q}$. Denote the quantum state corresponding to unit vector $w \in \C^{q}$ as $\ket{w} := \sum_{i \in [q]} w_{i} \ket{i}$. 
Let $k$ be the dimension of $W$ and let $w_{1}, w_{2}, \cdots, w_{k}$ be an orthonormal basis for $W$. Let 
$U_{W}$ be a unitary operation that maps the standard basis vectors $\ket{i} \to \ket{w_{i}}$. 
Then the measurement in $(W, W^{\perp})$ on state $\ket{\phi}$ can be implemented by first computing  
$U_{W}^{-1}\ket{\phi}$ and then measuring in the standard basis. The state $\ket{\phi}$ belongs to $W$ if and 
only if the result of measurement in the standard basis belongs to the set $\{1, 2, \cdots, k\}$. As 
the $(V, V^{\perp})$ measurement is on $\log q$ qubits, by 
Fact~\ref{fact:approx_unit}
we have, 

\begin{claim} 
\label{claim:measure-implement}
The measurement $(V,V^\perp)$ can be implemented to precision $1/q^{O(1)}$ 
in time $\widetilde{O}(q^{2})$.
\end{claim}

\noindent The implementation of the $(V, V^{\perp})$ measurement above shows that the sampling procedure can be performed efficiently.  
The procedure yields a sample from an $(A, \ell)$ distribution with $|A|=(q-r)+1$ 
when the measurement outcome is $V^{\perp}$. By Lemma \ref{lem:measure-bound} the outcome $V^{\perp}$ occurs with probability
$\Omega\left (q^{-c^{2}}\left( \frac{q}{r} \right )^{r} \right)$ if $(u,s)=r-1$. 
As $u$ is uniformly random on $\F_{q}^{n}$ at least a $\Omega\left (q^{-c^{2}-1}\left( \frac{q}{r} \right )^{r} \right)$ fraction of the samples are retained. We therefore have the following proposition, 
\begin{proposition}
\label{prop:sampling}
There is a quantum procedure that that runs in time $\widetilde{O}(q^{2})$, and given a pair $(u,\phi_s^u)$ 
where $u\in \Z_q^n$ is uniformly random and $\phi_s^u=\mu_{(u,s)}$, with probability at least
$\Omega\left (q^{-c^{2} -1} \left( \frac{q}{r} \right )^{r} \right )$ 
returns a sample from a $(A, \ell)$ distribution with $|A|=c+1$ and $\ell(x)=(s,x)$. 
\end{proposition}

\subsection{Recovering the secret $s$}

In order to solve $\HMS(q, n, r)$ given a scalar multiple $s_{0}=\lambda s$ found using Theorem~\ref{thm:LFS}, we need to find the scalar 
$\lambda$. We show that using $O(q)$ further 
input pairs we can find the value of $\lambda$ using 
a simple trial and error procedure given by the following lemma. 

\begin{lemma}
\label{lem:is_fixed_mu}
Given $t\in \F_q$ and a state $\tau\in\C^{r}$,
there is a quantum procedure that return YES
with probability 1 if $\tau=\mu_{t}$, while
if $\tau=\mu_{t'}$ for some $t'\in \F_q\setminus\{t\}$, it
returns YES with probability at most $p:=\begin{cases} 1/4 \text{ if $q< 3r/2$} \\ 
1- O(1/q) \text{ otherwise.}  \end{cases}$
\end{lemma}

\begin{proof}
Let $V_{t}$ be the one dimensional space spanned by $\mu_{t}$ and $V_{t}^{\perp}$ be its orthogonal complement. 
By Claim \ref{claim:measure-implement} we have an efficient implementation of the measurement in $(V_{t}, V_{t}^{\perp})$. 
If the state $\tau=\mu_{t}$ this measurement returns YES with probability $1$. The probability of returning YES for $\mu_{t'}$ is the squared inner product $|\mu_{t}^{\dagger}. \mu_{t'}|^{2}$.  We upper bound the probability of output YES for the 
two cases. 

Let $t' - t= \beta$ for a non zero $\beta \in F_{q}$.
Let $c=q-r$, the case $q<3r/2$ is equivalent to $q>3c$. For this case we have the estimate, 
\begin{eqnarray*}
|\mu_{t}^{\dagger}. \mu_{t'}|^{2}  = |\frac{1}{r}\sum_{h\in H}\omega^{h\beta}|^2 &=&
|\frac{1}{r}\sum_{h\in \F_q}\omega^{h\beta }-
\frac{1}{r}\sum_{h\in \F_q\setminus H}\omega^{h\beta }|^2 \\
&=&
|\frac{1}{r}\sum_{h\in \Z_q\setminus H}\omega^{h\beta }|^2 \\
&\leq &
\frac{c^2}{r^2}=\frac{c^2}{(q-c)^2}\leq \frac{c^2}{4c^2}=1/4. 
\end{eqnarray*}
If $q>3r/2$, then we use the following estimate, 
\begin{eqnarray*}
|\mu_{t}^{\dagger}. \mu_{t'}|^{2} =
|\frac{1}{r}\sum_{h\in H}\omega^{h\beta }|^2 
\leq
|\frac{1}{2}(1+e^{2\pi i/q})|^2
=1-O(1/q).
\end{eqnarray*} 
The above inequality follows as for $|H|>1$ the sum $\frac{1}{r}\sum_{h\in H}\omega^{h\beta}$ lies in the convex hull of the points $\zeta_{jk} := (\frac{\omega^{j} + \omega^{k}}{2})$ for $j, k \in [q], j< k$. 
The maximum norm for $\zeta_{jk}$ occurs when $k=j+1 \mod q$ and for this case the norm is $|\frac{1}{2}(1+e^{2\pi i/q})| $, the inequality follows. 

\end{proof}

\noindent Combining the results proved in this section with Algorithm \ref{alg:LFS}, we next obtain a quantum algorithm for $\HMS(q,n,r)$ via a reduction to 
$LFS(q, n, (q-r)+1)$ for the case of prime $q$. 

Theorem~\ref{thm:LFS} shows that given $N= \Omega ( \binom{n+d}{d} d^{2} \log (\binom{n+d}{d} )) $ samples from an $(A, \ell)$ distribution, 
a scalar multiple of the function $\ell$ can be found with constant probability. Proposition~\ref{prop:sampling} above shows that the expected time to obtain $N$ samples  
from the $(A, \ell)$ distribution is $\widetilde{O}(Nq^{c^{2} +3})$ where 
we used that each measurment requires time $\widetilde{O}(q^{2})$ and ignored the factor $\left( r/q\right )^{r} <1$. 
The number of samples and the time required for determining the scalar $\lambda$ in Lemma\ref{lem:is_fixed_mu} are negligible compared to these quantities. 
We therefore have the following theorem, 
\begin{theorem}  
\label{thm:small-prime2}
Let $q$ be a prime and let $c=q-r$.
Then there is a quantum algorithm which solves $\HMS(q,n,r)$
with sample complexity $\widetilde{O}( n^{c+1}q^{c^{2} +1})$ and in time $\widetilde{O}( n^{c+1}q^{c^{2} +3})$.
\end{theorem} 
\noindent The algorithm runs in time polynomial in $n, \log q$ for the case when $q=\poly(n)$. We therefore we have the following corollary, 
\begin{corollary} \label{cor:eff} 
Let $q=\poly(n)$ be a prime number and $c=q-r$ be a constant, then there is an efficient quantum algorithm for $\HMS(q,n,r)$. 
\end{corollary}

\section*{Acknowledgments}
A part of the research was accomplished while the first two authors were
visiting the Centre for Quantum Technologies (CQT), National University
of Singapore. The research at CQT was partially funded by the
Singapore Ministry of Education and the National Research Foundation
under grant R-710-000-012-135. This research was supported in part by the
QuantERA ERA-NET Cofund project QuantAlgo and by the  Hungarian
National Research, Development and Innovation Office -- NKFIH, Grant
K115288.

\bibliographystyle{alpha}
\bibliography{hmsp}

\newcommand{\etalchar}[1]{$^{#1}$}
\begin{thebibliography}{CKOR13}

\bibitem[Acc96]{vincenzo}
Vincenzo Acciaro.
\newblock The probability of generating some common families of finite groups.
\newblock {\em Utilitas Math.}, 49:243--254, 1996.

\bibitem[AG11]{AG}
Sanjeev Arora and Rong Ge.
\newblock New algorithms for learning in presence of errors.
\newblock In {\em Proceedings of the 38th International Colloquium on Automata,
  Languages and Programming {ICALP}, Zurich, Switzerland}, pages 403--415.
  Springer, 2011.

\bibitem[BKW03]{Blum}
Avrim Blum, Adam Kalai, and Hal Wasserman.
\newblock Noise-tolerant learning, the parity problem, and the statistical
  query model.
\newblock {\em Journal of the ACM (JACM)}, 50(4):506--519, 2003.

\bibitem[CKOR13]{ChildsKOR13}
Andrew~M. Childs, Robin Kothari, Maris Ozols, and Martin Roetteler.
\newblock Easy and hard functions for the boolean hidden shift problem.
\newblock In Simone Severini and Fernando G. S.~L. Brand{\~{a}}o, editors, {\em
  8th Conference on the Theory of Quantum Computation, Communication and
  Cryptography, {TQC} 2013, May 21-23, 2013, Guelph, Canada}, volume~22 of {\em
  LIPIcs}, pages 50--79. Schloss Dagstuhl - Leibniz-Zentrum fuer Informatik,
  2013.

\bibitem[CvD07]{CD07}
Andrew~M. Childs and Wim van Dam.
\newblock Quantum algorithm for a generalized hidden shift problem.
\newblock In Nikhil Bansal, Kirk Pruhs, and Clifford Stein, editors, {\em
  Proceedings of the Eighteenth Annual {ACM-SIAM} Symposium on Discrete
  Algorithms, {SODA} 2007, New Orleans, Louisiana, USA, January 7-9, 2007},
  pages 1225--1232. {SIAM}, 2007.

\bibitem[DIK{\etalchar{+}}14]{DIKQS14}
Thomas Decker, G{\'a}bor Ivanyos, Raghav Kulkarni, Youming Qiao, and Miklos
  Santha.
\newblock An efficient quantum algorithm for finding hidden parabolic subgroups
  in the general linear group.
\newblock In {\em International Symposium on Mathematical Foundations of
  Computer Science}, pages 226--238. Springer, 2014.

\bibitem[EH00]{EH00}
Mark Ettinger and Peter H{\o}yer.
\newblock On quantum algorithms for noncommutative hidden subgroups.
\newblock {\em Advances in Applied Mathematics}, 25:239--251, 2000.

\bibitem[FIM{\etalchar{+}}14]{Friedl}
Katalin Friedl, G{\'a}bor Ivanyos, Fr{\'e}d{\'e}ric Magniez, Miklos Santha, and
  Pranab Sen.
\newblock Hidden translation and translating coset in quantum computing.
\newblock {\em SIAM Journal on Computing}, 43(1):1--24, 2014.
\newblock preliminary version in STOC 2003.

\bibitem[GRR11]{GavinskyRR11}
Dmitry Gavinsky, Martin Roetteler, and J{\'{e}}r{\'{e}}mie Roland.
\newblock Quantum algorithm for the boolean hidden shift problem.
\newblock In Bin Fu and Ding{-}Zhu Du, editors, {\em Computing and
  Combinatorics - 17th Annual International Conference, {COCOON} 2011, Dallas,
  TX, USA, August 14-16, 2011. Proceedings}, volume 6842 of {\em Lecture Notes
  in Computer Science}, pages 158--167. Springer, 2011.

\bibitem[Had93]{H93}
Jacques Hadamard.
\newblock Resolution d'une question relative aux determinants.
\newblock {\em Bull. des Sciences Math.}, 2:240--246, 1893.

\bibitem[Iva08]{Ivanyos}
G{\'{a}}bor Ivanyos.
\newblock On solving systems of random linear disequations.
\newblock {\em Quantum Information {\&} Computation}, 8(6):579--594, 2008.

\bibitem[Kal85]{Kaltofen85}
Erich Kaltofen.
\newblock Polynomial-time reductions from multivariate to bi- and univariate
  integral polynomial factorization.
\newblock {\em {SIAM} J. Comput.}, 14(2):469--489, 1985.

\bibitem[Kup05]{Kup05}
Greg Kuperberg.
\newblock A subexponential-time quantum algorithm for the dihedral hidden
  subgroup problem.
\newblock {\em SIAM Journal on Computing}, 35(1):170--188, 2005.

\bibitem[MRRS04]{MRRS04}
Cristopher Moore, Daniel Rockmore, Alexander Russell, and Leonard~J Schulman.
\newblock The power of basis selection in {Fourier} sampling: Hidden subgroup
  problems in affine groups.
\newblock In {\em Proceedings of the 15th annual Symposium on Discrete
  algorithms}, pages 1113--1122. Society for Industrial and Applied
  Mathematics, 2004.

\bibitem[NC02]{NC02}
Michael~A Nielsen and Isaac Chuang.
\newblock Quantum computation and quantum information, 2002.

\bibitem[Pei09]{Peikert}
Chris Peikert.
\newblock Public-key cryptosystems from the worst-case shortest vector problem.
\newblock In {\em Proceedings of the forty-first annual ACM symposium on Theory
  of computing}, pages 333--342. ACM, 2009.

\bibitem[Pom02]{Pomerance}
Carl Pomerance.
\newblock The expected number of random elements to generate a finite abelian
  group.
\newblock {\em Periodica Mathematica Hungarica}, 43(1):191--198, Aug 2002.

\bibitem[Reg09]{Regev}
Oded Regev.
\newblock On lattices, learning with errors, random linear codes, and
  cryptography.
\newblock {\em Journal of the ACM (JACM)}, 56(6):34, 2009.

\bibitem[Roe16]{Roetteler16}
Martin Roetteler.
\newblock Quantum algorithms for abelian difference sets and applications to
  dihedral hidden subgroups.
\newblock In Anne Broadbent, editor, {\em 11th Conference on the Theory of
  Quantum Computation, Communication and Cryptography, {TQC}, September 27-29,
  2016, Berlin, Germany}, volume~61 of {\em LIPIcs}, pages 8:1--8:16. Schloss
  Dagstuhl - Leibniz-Zentrum fuer Informatik, 2016.

\bibitem[Sch80]{S80}
Jacob~T Schwartz.
\newblock Fast probabilistic algorithms for verification of polynomial
  identities.
\newblock {\em Journal of the ACM (JACM)}, 27(4):701--717, 1980.

\bibitem[vzGK83]{GK83}
Joachim von~zur Gathen and Erich Kaltofen.
\newblock Polynomial-time factorization of multivariate polynomials over finite
  fields.
\newblock In {\em International Colloquium on Automata, Languages, and
  Programming}, pages 250--263. Springer, 1983.

\bibitem[Zip79]{Z79}
Richard Zippel.
\newblock Probabilistic algorithms for sparse polynomials.
\newblock {\em Symbolic and algebraic computation}, pages 216--226, 1979.

\end{thebibliography}

\end{document}